%% file: q-ary-mac.tex
\documentclass[journal,letterpaper,final,onecolumn]{IEEEtran} %!PN
\usepackage{amsmath}
\usepackage{amsfonts}
\usepackage[utf8]{inputenc}
\usepackage{url}
\usepackage{graphicx}

\DeclareMathOperator{\co}{co}
\newcommand{\cF}{{\mathcal F}}
\newcommand{\cI}{{\mathfrak I}}
\newcommand{\cJ}{{\mathcal K}}
\newcommand{\cR}{{\mathfrak R}}

\newcommand{\cX}{{\mathcal X}}
\newcommand{\cY}{{\mathcal Y}}

\newcommand{\cW}{{\mathcal W}}

\newcommand{\cB}{{\mathcal B}}
\newcommand{\cC}{{\mathfrak C}}
\newcommand{\fF}{{\mathbb F}}
\newcommand{\fR}{{\mathbb R}}
\newcommand{\bs}{{\mathbf s}}
\def\paren#1{\smash{(}#1\smash{)}}

\newtheorem{proposition}[equation]{Proposition}
\newtheorem{corollary}[equation]{Corollary}
\newtheorem{theorem}[equation]{Theorem}

\newtheorem{remark}[equation]{Remark}

\newtheorem{lemma}[equation]{Lemma}

\newtheorem{definition}[equation]{Definition}

\DeclareMathOperator{\card}{\#}

\begin{document}
\title{Polar codes for the two-user multiple-access channel}
\author{Eren \c{S}a\c{s}o\u{g}lu, Emre Telatar, Edmund Yeh}
\maketitle

\begin{abstract}
\noindent
Ar{\i}kan's polar coding method is extended to two-user multiple-access
channels.  It is shown that if the two users of the channel use the Ar{\i}kan
construction, the resulting channels will polarize to one of five
possible extremals, on each of which uncoded transmission is optimal.
The sum rate achieved by this coding technique is the one that correponds
to uniform input distributions.  The encoding and decoding complexities
and the error performance of these codes are as in the single-user
case: $O(n\log n)$ for encoding and decoding, and $o(\exp(-n^{1/2-\epsilon}))$
for block error probability, where $n$ is the block length.
\end{abstract}
%\begin{keywords} Capacity-achieving codes, channel polarization, polar codes.
%\end{keywords}

\section{Introduction}

Polar coding, invented by Ar{\i}kan~\cite{Arikan2009}, is a
technique for 
achieving
the `symmetric capacity' of binary input, memoryless channels.  The
underlying principle of the technique is to convert repeated uses of a
given single-user channel to single uses of a set of extremal
channels---almost every channel in the set is either almost perfect, or
almost useless. Ar{\i}kan calls this phenomenon \emph{polarization}.  In this
note we describe a way to extend this technique to multiple-access
channels (MACs).

One way to do this extension is via the `rate splitting/onion peeling' scheme
of~\cite{RimUrb1996}, \cite{GRUW2001}. In Appendix A, we
describe how arbitrary points in the capacity region of a given MAC
can be achieved using polar codes and rate splitting techniques.

The approach taken in here is different, partly because
our motivation is to see whether multiple-access channels
polarize in the same way as single-user channels do.  In the following,
we will describe a technique to `polarize' a given two-user
multiple-access channel in the same sense as
in~\cite{Arikan2009}, i.e., we
will convert repeated uses of this MAC into single uses of extremal MACs.
Whereas in the single user case there are only two extremal channels
(perfect or useless), we will see that in the multiple-access case
there will be five.

The coding scheme that results from this construction shares some
properties of the single-user case: the encoding and decoding
complexity is $O(n\log n)$, $n$ being the block length, and the
block error probability is roughly $O(2^{-\sqrt{n}})$. Also analogous
to the single-user polar codes' achieving the `symmetric capacity', codes for
the multiple-access channel are capable of
achieving some of the rate pairs on the dominant face of the rate
region obtained with uniformly distributed inputs.

\section{Preliminaries}
\label{sec:preliminaries}
Let $P\colon\cX\times\cW\to\cY$ be a two user multiple-access channel with
input alphabets $\cX=\cW=\fF_q=\{0,1,\dotsc,q-1\}$, where $q$ is a prime
number. The output alphabet $\cY$ may be arbitrary. The channel is
specified by $P(y|x,w)$, the conditional probability
of each output symbol $y\in\cY$ for each possible input symbol pair
$(x,w)\in\cX\times\cW$.
The capacity region of such a channel is given by
$$
\cC(P):=\co\Bigl(\bigcup_{X,W}\cR(X,W)\Bigr)
$$
where
\begin{align*}
\cR(W,X)=\bigr\{(R_1,R_2)\colon 0\leq R_1&\leq I(X;YW),\;\\
	0\leq R_2&\leq I(W;YX),\;\\
	R_1+R_2&\leq I(XW;Y)\bigr\},
\end{align*}
the union is over all random variables $X\in\cX$, $W\in\cW$, and $Y\in\cY$ jointly distributed as
$$
p_{XWY}(x,w,y)=p_X(x)p_W(w)P(y\mid x,w),
$$
and $\co(S)$ denotes the convex hull of the set $S$\footnote{All
logarithms in this note will be to the base $q$. This in
particular implies that $I(X;YW),I(W;YX)\in[0,1]$ and
$I(XW;Y)\in[0,2]$.}
.

In this note, rather than the capacity region, we will be interested in
the region 
\begin{align*}
\cI(P)&:=\cR(X,W)\quad
\text{where $X$ and $W$ are \emph{uniformly distributed}
	on $\fF_q$}.
\end{align*}

Given such a channel $P$ and independent random variables $X,W$ uniformly
distributed on $\fF_q$, define 
$$
I^{\paren{1}}(P):=I(X;YW),\quad I^{\paren{2}}(P):=I(W;YX), \text{ and } I^{\paren{12}}(P):=I(XW;Y),
$$
and let
$$
\cJ(P):=\big(I^{\paren{1}}(P),I^{\paren{2}}(P),I^{\paren{12}}(P)\big) \in \mathbb{R}^3.
$$
Note that the region $\cI(P)$ is defined by
$$
\cI(P)=\bigl\{(R_1,R_2)\colon 0\leq R_1\leq I^{\paren{1}}(P), \;
	0\leq R_2\leq I^{\paren{2}}(P),\; R_1+R_2\leq I^{\paren{12}}(P)\bigr\}.
$$
Further note that $\max\{I^{\paren{1}},\,I^{\paren{2}}\}
\leq I^{\paren{12}}\leq I^{\paren{1}}+I^{\paren{2}}$, therefore the constraints
that define $\cI(P)$ are polymatroidal.  In particular, there exists
$(R_1,R_2)\in\cI(P)$ for which $R_1+R_2=I^{\paren{12}}$. The set of such points is called the \emph{dominant face} of $\cI(P)$.

\section{Polarization}
\label{sec:polarization}
Two independent uses of $P$ yields a multiple-access channel $P^2$ with
input alphabets $\cX^2$ and $\cW^2$, and output alphabet $\cY^2$.  Setting
the inputs $(X_1,X_2,W_1,W_2)$ to be independent and uniformly
distributed on $\fF_q$, and letting $(Y_1,Y_2)$ denote the output,
the region  $\cI(P^2)$ is described by the three quantities
\begin{align*}
I(X_1X_2;Y_1Y_2W_1W_2) &= 2I^{\paren{1}}(P),\\
I(W_1W_2;Y_1Y_2X_1X_2) &= 2I^{\paren{2}}(P),\quad\text{and}\\
I(X_1X_2W_1W_2;Y_1Y_2) &= 2I^{\paren{12}}(P)
\end{align*}
that upper bound $R_1$, $R_2$ and $R_1+R_2$ respectively.  Now
consider putting the pair $(X_1,X_2)\in\fF_q^2$ in one-to-one correspondence
with $(U_1,U_2)\in\fF_q^2$ via
$$
X_1=U_1+U_2,\quad X_2=U_2
$$
and the pair $(W_1,W_2)\in\fF_q^2$ in one-to-one correspondence
with $(V_1,V_2)\in\fF_q^2$ via 
$$
W_1=V_1+V_2,\quad W_2=V_2,
$$
where both additions are modulo-$q$. Observe that $(U_1,U_2,V_1,V_2)$ are also independent and uniformly distributed
on $\fF_q$.  Note further that
\begin{align}
\label{eq:s1}
2I^{\paren{1}}(P)&=I(X_1X_2;Y_1Y_2W_1W_2)\\
&=I(U_1U_2;Y_1Y_2V_1V_2)\notag\\
&=I(U_1;Y_1Y_2V_1V_2)+I(U_2;Y_1Y_2V_1V_2U_1)\notag\\
&\geq I(U_1;Y_1Y_2V_1)+I(U_2;Y_1Y_2V_1V_2U_1),\notag\\[1.5ex]
\label{eq:s2}
2I^{\paren{2}}(P)&=I(W_1W_2;Y_1Y_2X_1X_2)\\
&=I(V_1V_2;Y_1Y_2U_1U_2)\notag\\
&=I(V_1;Y_1Y_2U_1U_2)+I(V_2;Y_1Y_2U_1U_2V_1)\notag\\
&\geq I(V_1;Y_1Y_2U_1)+I(V_2;Y_1Y_2U_1U_2V_1),\notag\\[-1.5ex]
\intertext{and}
\label{eq:s12}
2I^{\paren{12}}(P)&=I(X_1X_2W_1W_2;Y_1Y_2)\\
&=I(U_1U_2V_1V_2;Y_1Y_2)\notag\\
&=I(U_1V_1;Y_1Y_2)+I(U_2V_2;Y_1Y_2U_1V_1).\notag
\end{align}

Observe that the quantities
$$
I(U_1;Y_1Y_2V_1),\quad
I(V_1;Y_1Y_2U_1),\quad\text{and}\quad
I(U_1V_1;Y_1Y_2)
$$
are those that describe the region associated to 
the $q$-ary input multiple-access channel $U_1V_1\to Y_1Y_2$, and the
quantities
$$
I(U_2;Y_1Y_2V_1V_2U_1),\quad
I(V_2;Y_1Y_2U_1U_2V_1),\quad\text{and}\quad
I(U_2V_2;Y_1Y_2U_1V_1)
$$
are those that describe the region associated to
the $q$-ary input multiple-access channel $U_2V_2\to Y_1Y_2U_1V_1$.  This
motivates the following.

\begin{definition}
Suppose $P\colon\cX\times\cW\to\cY$ is a two user multiple-access channel
with input alphabet $\fF_q$.  Define two new multiple-access channels,
$P^-\colon\cX\times\cW\to\cY\times\cY$ and
$P^+\colon\cX\times\cW\to\cY\times\cY\times\cX\times\cW$ as:
\begin{align*}
P^{-}(y_1,y_2|u_1,v_1)&=\sum_{u_2\in\cX}\sum_{v_2\in\cW}
	\tfrac1{q^2} P(y_1|u_1+u_2,v_1+v_2) P(y_2|u_2,v_2),\\
P^{+}(y_1,y_2,u_1,v_1|u_2,v_2)&=\tfrac1{q^2} P(y_1|u_1+u_2,v_1+v_2) P(y_2|u_2,v_2).
\end{align*}
\end{definition}

The channels $P^-$ and $P^+$ correspond to $U_1V_1\to Y_1Y_2$ and $U_2V_2\to
Y_1Y_2U_1V_1$ above, respectively. It is clear that the channel $P^-$ can be \emph{synthesized} from two
independent uses of the channel $P$, whereas the channel $P^+$ in general cannot,
since at its output we require $(U_1,V_1)$ in addition to $(Y_1,Y_2)$.
However, $P^+$ can be synthesized from two uses of the channel $P$ with
the aid of a \emph{genie} that delivers $(U_1,V_1)$ as side information to
the output terminal.

Note that the channel $P^-$ is `worse' and the channel $P^+$ is `better'
than the channel $P$ in the sense that
$I^{\paren{\alpha}}(P^-)\leq I^{\paren{\alpha}}(P)\leq I^{\paren{\alpha}}(P^+)$
for each $\alpha\in\{1,2,12\}$.  To see this, observe that if we process
the output $(y_1,y_2,u_1,v_1)$ of the channel $P^+$ to keep only $y_2$,
the resulting channel is identical to $P$.  Thus
$I^{\paren{\alpha}}(P^+)\geq I^{\paren{\alpha}}(P)$.  That
$I^{\paren{\alpha}}(P^-)\leq I^{\paren{\alpha}}(P)$ then follows from~\eqref{eq:s1}, \eqref{eq:s2} and~\eqref{eq:s12}.
Consequently,
$$
\cI(P^-) \subset \cI(P) \subset \cI(P^+).
$$
Furthermore, by virtue of equations~\eqref{eq:s1}, \eqref{eq:s2}
and~\eqref{eq:s12},
$$
\tfrac12\cI(P^-)+\tfrac12\cI(P^+)\subset \cI(P),
$$
where the left-hand side of the above denotes set sum, i.e.,
$$
\tfrac12\cI(P^-)+\tfrac12\cI(P^+) = \left\{ \tfrac12 a + \tfrac12 b\colon a\in\cI(P^-), b\in\cI(P^+)\right\}.
$$
Nevertheless, by the polymatroidal nature of
$(I^{\paren{1}},I^{\paren{2}},I^{\paren{12}})$ and
by~\eqref{eq:s12}, there are points in $\frac12\cI(P^-)+\frac12\cI(P^+)$
that are on the dominant face of $\cI(P)$.

We have now seen that from two independent copies of a $q$-ary input
multiple-access channel $P$ we can derive two $q$-ary input multiple-access
channels $P^-$ and $P^+$.  Applying the same process to $P^-$ and
$P^+$, we can derive from four independent copies of $P$, four $q$-ary
input multiple-access channels $P^{--}:=(P^-)^-$, $P^{-+}:=(P^-)^+$,
$P^{+-}:=(P^+)^-$ and $P^{++}:=(P^+)^+$.  Recursively applying the process
$\ell$ times results in $2^\ell$ $q$-ary input multiple-access
channels
$$
P^{-\dotsm-},\dots,P^{+\dotsm+}.
$$
These channels have the property that the set
$$
2^{-\ell} \sum_\bs \cI(P^{\bs})
$$
is a subset of $\cI(P)$, but contains points on the dominant face of $\cI(P)$.

The main result reported in this section is that these derived channels polarize
in the following sense:
\begin{theorem}
\label{thm:polar}
Let $P$ be a $q$-ary input multiple-access channel. Let
$M:=\{(0,0,0),\,(0,1,1),\,(1,0,1),\,(1,1,1),\,(1,1,2)\}\subset\fR^3$, and
for $p\in\fR^3$, let $d\bigl(p,M\bigr):=\min_{x\in M}\|p-x\|$
denote the distance from a point $p$ to $M$.
Then, for any $\delta>0$
$$
\lim_{\ell\to\infty}\frac{1}{2^\ell} \card\bigl\{\bs\in\{-,+\}^\ell\colon
	d\bigl( \cJ(P^\bs),M\bigr)\geq\delta
\bigr\}=0.
$$
That is, except for a vanishing fraction, the regions $\cI(P^\bs)$
approach one of five possible regions.
\end{theorem}

\begin{remark}
\label{rmk:limiting}
The five limiting regions in Theorem~\ref{thm:polar}
 are the following.

\centerline{\input{macs}}

The first case, (000), is that of a channel whose output provides
no useful information about any of its inputs; the second and
third, (011) and (101), are channels that provide complete information
about one of the inputs but nothing about the other; the fourth,
(111), is a pure contention channel; the last, (112), is one whose output
determines both inputs perfectly.
\end{remark}

Theorem~\ref{thm:polar} will be proved as a corollary to
Theorem~\ref{thm:limit} below. To prove the latter theorem, we need a few auxiliary results:

\begin{lemma}[\cite{Sasoglu2010}]
\label{lem:basic}
For any $\epsilon>0$, there is a $\delta:=\delta(\epsilon)>0$ such that if
\begin{itemize}
\item[(i)]
$Q\colon\fF_q\to\cB$ is a $q$-ary input channel with arbitrary output alphabet $\cB$, and
\item[(ii)]
$A_1,A_2,B_1,B_2$ are random variables jointly distributed as
$$
p_{A_1A_2B_1B_2}(a_1,a_2,b_1,b_2)=\tfrac1{q^2}Q(b_1|a_1+a_2)Q(b_2|a_2),\quad\text{and}
$$
\item[(iii)]
$I(A_2;B_1B_2A_1)-I(A_2;B_2)<\delta$,
\end{itemize}
then,
$$
I(A_2;B_2)\notin(\epsilon,1-\epsilon).
$$
Note that $\delta$ can be chosen irrespective of the alphabet $\cB$.
\end{lemma}

\begin{corollary}
\label{cor:i1}
For any $\epsilon>0$ there exists a $\delta>0$ such that if $P$ is
a two-user $q$-ary input multiple-access channel with
$$
I^{\paren{1}}(P^+)-I^{\paren{1}}(P)<\delta,
$$
then, $I^{\paren{1}}(P)\notin(\epsilon,1-\epsilon)$. Similarly, if $P$ is such that 
$$
I^{\paren{2}}(P^+)-I^{\paren{2}}(P)<\delta,
$$
then, $I^{\paren{2}}(P)\notin(\epsilon,1-\epsilon)$. 
\end{corollary}
\begin{proof}
It suffices to prove the first claim. To that end, note that
$I^{\paren{1}}(P)=I(U_2;Y_2V_2)$ and
$I^{\paren{1}}(P^+)=I(U_2;Y_1Y_2U_1V_1V_2)$,
where
$$
p_{U_1V_1U_2V_2Y_1Y_2}(u_1,v_1,u_2,v_2,y_1,y_2)=\tfrac{1}{q^4} P(y_1\mid
u_1+u_2,v_1+v_2)P(y_2\mid u_2,v_2).
$$
We have by hypothesis that
$$
\delta>%I^{\paren{1}}(P^+)-I^{\paren{1}}(P)
I(U_2;Y_1Y_2V_1V_2U_1)-I(U_2;Y_2V_2).
$$
It can easily be checked that the values of the above mutual informations remain
unaltered if evaluated under the joint distribution
$$
q_{U_1V_1U_2V_2Y_1Y_2}(u_1,v_1,u_2,v_2,y_1,y_2)=\tfrac{1}{q^4} P(y_1\mid
u_1+u_2,v_1)P(y_2\mid u_2,v_2).
$$
Defining $A_i=U_i$, $B_i=(Y_i,V_i)$ and
$Q(y,v|u)=\tfrac12 P(y|u,v)$, one can then write
$$
q_{A_1A_2B_1B_2}(a_1,a_2,b_1,b_2) = \tfrac1{q^2} Q(b_1\mid a_1+a_2)Q(b_2\mid a_2).
$$
Applying Lemma~\ref{lem:basic} now yields the claim.
\end{proof}

\begin{lemma}
\label{lem:i12}
For any $\epsilon>0$ there exists a $\delta>0$ such that whenever $P$ is
a two-user $q$-ary input multiple-access channel with $I^{\paren{12}}(P^+)-I^{\paren{12}}(P)<\delta$, then
$$
I^{\paren{12}}(P) - I^{\paren{j}}(P) \notin(\epsilon,1-\epsilon) \quad \text{ for } j=1,2.
$$
\end{lemma}
\begin{proof}
By symmetry, it suffices to prove the claim for $j=1$. Choose $\delta$ so that $\delta<\epsilon$ and $\delta<\delta(\epsilon)$ of
Lemma~\ref{lem:basic}.
Note that
\begin{align*}
\delta &> I^{\paren{12}}(P^+)-I^{\paren{12}}(P)\\
&= I(U_2V_2;Y_1Y_2U_1V_1)-I(U_2V_2;Y_2)\\
&= I(U_2V_2;Y_1U_1V_1|Y_2)\\
&\geq I(U_2;Y_1U_1|Y_2)\\
&=I(U_2;Y_1Y_2U_1)-I(U_2;Y_2).
\end{align*}
Applying Lemma~\ref{lem:basic} with $A_i=U_i$, $B_i=Y_i$ and
$Q(y|u)=\sum_v \tfrac1q P(y|u,v)$ we conclude that
$I(U_2;Y_2)\notin(\epsilon,1-\epsilon)$.
Since $I(U_2;Y_2)=I^{\paren{12}}(P)-I^{\paren{1}}(P)$, the claim follows.
\end{proof}

Suppose $P$ is a two-user $q$-ary input MAC.  Let $B_1,B_2,\dots$ be an i.i.d.\ sequence of random
variables taking values in the set $\{-,+\}$, with
$\Pr(B_1=\mathord{-})=\Pr(B_1=\mathord{+})=1/2$.
Define a \emph{MAC-valued} random process $\{P_\ell\colon \ell\geq0\}$ via
\begin{align}
\label{eqn:process-P}
P_0:=P,\quad P_\ell:=P_{\ell-1}^{B_\ell},\: \ell\geq 1.
\end{align}
Further define random processes $\{I^{\paren{1}}_\ell\colon \ell\geq0\}$,
$\{I^{\paren{2}}_\ell\colon \ell\geq0\}$ and $\{I^{\paren{12}}_\ell\colon \ell\geq0\}$ as
\begin{align}
\label{eqn:process-I}
I^{\paren{1}}_\ell:=I^{\paren{1}}(P_\ell),\quad
I^{\paren{2}}_\ell:=I^{\paren{2}}(P_\ell),\quad\text{and}\quad
I^{\paren{12}}_\ell:=I^{\paren{12}}(P_\ell).
\end{align}

\begin{lemma}
The processes $\{I^{\paren{1}}_\ell\colon \ell\geq0\}$ and $\{I^{\paren{2}}_\ell\colon \ell\geq0\}$
are bounded supermartingales, the process $\{I^{\paren{12}}_\ell\colon \ell\geq0\}$
is a bounded martingale.
\end{lemma}
\begin{proof}
Since $P_\ell$ is a $q$-ary input MAC, $I^{\paren{1}}_\ell$ and $I^{\paren{2}}_\ell$ take values is
$[0,1]$ and $I^{\paren{12}}_\ell$ takes values in $[0,2]$, and thus the processes are
bounded.  The martingale claims follow from~\eqref{eq:s1}, \eqref{eq:s2}
and~\eqref{eq:s12} respectively.
\end{proof}

\begin{theorem}
\label{thm:limit}
The process $(I^{\paren{1}}_\ell,I^{\paren{2}}_\ell,I^{\paren{12}}_\ell)$ converges almost surely,
and the limit
$$
(I^{\paren{1}}_\infty,I^{\paren{2}}_\infty,I^{\paren{12}}_\infty):=
	\lim_{\ell\to\infty}(I^{\paren{1}}_\ell,I^{\paren{2}}_\ell,I^{\paren{12}}_\ell)
$$
belongs to the set $\{(0,0,0),(0,1,1), (1,0,1), (1,1,1), (1,1,2)\}$
with probability $1$.
\end{theorem}
\begin{proof}
Let $(\Omega,\Pr,\cF)$ be the probability space these processes are
defined in.  Let
$$
A:=\bigl\{\omega\in\Omega\colon\text{$\lim_{\ell\to\infty}
	I^{\paren{\alpha}}_\ell$ exists for each $\alpha\in\{1,2,12\}$}\bigr\}.
$$
The almost sure convergence of $I^{\paren{1}}_\ell$ and $I^{\paren{2}}_\ell$ follow
from them being bounded supermartingales, almost sure convergence of
$I^{\paren{12}}_\ell$ follows form it being a bounded martingale.  Thus, $\Pr(A)=1$,
and it remains to show that the joint limit belongs to the set claimed.

To that end we will first show that $I^{\paren{1}}_\infty\in\{0,1\}$ a.s.  Since
$I^{\paren{1}}_\ell$ converges a.s., $\lim_{\ell\to\infty}
|I^{\paren{1}}_{\ell+1}-I^{\paren{1}}_\ell|=0$ a.s.  Since
$|I^{\paren{1}}_{\ell+1}-I^{\paren{1}}_\ell|$ is bounded (by 1), it follows that
$\lim_{\ell\to\infty} E\bigl[|I^{\paren{1}}_{\ell+1}-I^{\paren{1}}_\ell|\bigr]=0$.  But
$$
E\bigl[|I^{\paren{1}}_{\ell+1}-I^{\paren{1}}_\ell|\bigr]\geq
	\tfrac12\bigl[I^{\paren{1}}(P_\ell^+)-I^{\paren{1}}(P_\ell)\bigr],
$$
and we see that $\lim_{\ell\to\infty} I^{\paren{1}}(P_\ell^+)-I^{\paren{1}}(P_\ell)=0$.
From Corollary~\ref{cor:i1} we conclude that
$\lim_{\ell\to\infty}I^{\paren{1}}_\ell\in\{0,1\}$.

Swapping the roles of the two users yields $I^{\paren{2}}_\infty\in\{0,1\}$ a.s.  We
thus find that $(I^{\paren{1}}_\infty,I^{\paren{2}}_\infty)$ is equal to either $(0,0)$,
$(0,1)$, $(1,0)$, or $(1,1)$.  Denoting the set of $\omega\in A$ for which
$I^{\paren{1}}_\infty=a$, $I^{\paren{2}}_\infty=b$ by $A_{ab}$, we see that
$A=A_{00}\cup A_{01}\cup A_{10}\cup A_{11}$.  Since
$$
\max\{I^{\paren{1}},I^{\paren{2}}\}\leq I^{\paren{12}}\leq I^{\paren{1}}+I^{\paren{2}},
$$
we conclude that the value of $I^{\paren{12}}_\infty$ in $A_{00}$, $A_{01}$ and $A_{10}$ is $0$, $1$, and $1$ respectively.

All that remains now is to show that $I^{\paren{12}}_\infty$ belongs to $\{1,2\}$
for $\omega\in A_{11}$.  To that end note that for any $\omega\in A_{11}$,
\begin{enumerate}
\item[(i)] $\lim_{\ell\to\infty} I^{\paren{1}}(P_\ell)=1$,
$\lim_{\ell\to\infty} I^{\paren{2}}(P_\ell)=1$, and
\item[(ii)] $\lim_{\ell\to\infty} I^{\paren{12}}(P_\ell)$ exists and thus
$\lim_{\ell\to\infty}\bigl|I^{\paren{12}}(P_{\ell+1})-I^{\paren{12}}(P_\ell)\bigr|=0$.
\end{enumerate}
But
$$
\bigl|I^{\paren{12}}(P_{\ell+1})-I^{\paren{12}}(P_\ell)\bigr|=I^{\paren{12}}(P_\ell^+)-I^{\paren{12}}(P_\ell),
$$
and thus $\lim_{\ell\to\infty} I^{\paren{12}}(P_\ell^+)-I^{\paren{12}}(P_\ell)=0$.
Now Lemma~\ref{lem:i12} lets us conclude that
$\lim_{\ell\to\infty}I^{\paren{12}}(P_\ell)\in\{1,2\}$.
\end{proof}

\begin{newproof}{Proof of Theorem~\ref{thm:polar}}
When the processes $P_\ell$ and
$(I_\ell^{\paren{1}},I_\ell^{\paren{2}},I_\ell^{\paren{12}})$, $\ell=0,1,\dotsc$  are defined as in \eqref{eqn:process-P} and \eqref{eqn:process-I}, respectively, we have 
$$
\Pr\left[d((I_\ell^{\paren{1}},I_\ell^{\paren{2}},I_\ell^{\paren{12}}),M)\ge\delta\right] =
\frac1{2^\ell}\#\left\{\bs\in\{-,+\}^\ell\colon
d(\cJ(P^\bs),M)\ge\delta\right\}.
$$
The claim then follows from Theorem~\ref{thm:limit}.
\end{newproof}

\section{Rate of polarization}
We have seen that any $q$-ary input MAC can be polarized to a set of
five extremal MACs, by recursively applying the channel
combining/splitting procedure of Section~\ref{sec:polarization}. Furthermore,
Remark~\ref{rmk:limiting} suggests a natural scheme to exploit this
phenomenon---polar coding \cite{Arikan2009}: one can hope to
communicate reliably by sending uncoded information over the reliable
channels, and not sending any information over the others. In this section,
we will formalize this intuition, showing that such a coding scheme achieves points on the dominant
face of $\cI(P)$.

We first introduce some notation: Given a $q$-ary input multiple-access channel
$P$, define two point-to-point channels
$P[U]\colon\mathbb{F}_q\to\cY$ and $P[U\mid
V]\colon\mathbb{F}_q\to\cY\times\fF_q$ through
$$
P[U](y\mid u) = \tfrac1q \sum_v P(y\mid u,v)
$$
$$
P[U\mid V](y,v\mid u) = \tfrac1q P(y\mid u,v)
$$
That is, $P[U]$ is the channel $U\to Y$, and $P[U\mid V]$ is the
channel $U\to YV$. Define $P[V]$ and $P[V\mid U]$ analogously.
Also, for every $\alpha,\gamma\in\fF_q$ define the channels $P[\alpha,\gamma]\colon\fF_q\to\cY$
through
$$
P[\alpha,\gamma](y\mid s) = \tfrac1q \sum_{\substack{u,v\colon\\
\alpha u+\gamma
v=s}} P(y\mid u,v) 
$$
That is, $P[\alpha,\gamma]$ is the channel 
$\alpha U+\gamma
V\to Y$. 

Given a point-to-point channel $Q\colon\fF_q\to\cY$, let $P_e(Q)$ denote its average
probability of error with uniform input distribution and the optimal
(ML) decision rule. Also let $I(Q)$ denote the mutual information developed
across $Q$ with uniform inputs. That is, 
$$
I(Q)= \tfrac1q \sum_{x,y} Q(y\mid x) \log \frac{Q(y\mid
x)}{\tfrac1q\sum_{x'} Q(y\mid x')}.
$$
Finally let $Z(Q)$ define the
\emph{Bhattacharyya parameter} of $Q$, defined as
$$
Z(Q)=\tfrac1{q(q-1)}\sum_{x\neq x'}\sum_{y\in\cB}
\sqrt{Q(y\mid x)Q(y\mid x')}.
$$
It is known (see \cite{STA2009}) that $P_e(Q)\le qZ(Q)$.

We are now ready to describe the encoding rule: Fix $\ell$ and let $n=2^\ell$. Let $B_n$ denote the $n\times n$ permutation matrix called the `bit reversal' operator in \cite{Arikan2009}, and let 
$G_n = \left[\begin{smallmatrix} 1 & 0 \\ 1 & 1\end{smallmatrix}\right]^{\otimes\ell}$ denote the $\ell$th Kronecker power of the matrix
$\left[\begin{smallmatrix} 1 & 0 \\ 1 & 1\end{smallmatrix}\right]$.
Put $U^n:=(U_1,\dots,U_n)$ and $V^n=(V_1,\dots,V_n)$ into one-to-one 
correspondence with $X^n=(X_1,\dots,X_n)$ and $W^n=(W_1,\dots,W_n)$ via 
\begin{align*}
X^n &=U^nB_n G_n, \\
W^n &=V^nB_n G_n.
\end{align*}
Transmit
$(X^n,W^n)$ over $n$ independent uses of $P$ and receive $Y^n$.
Defining
$P_{(i)}\colon\cX\times\cW\to\cY^n\times\cX^{i-1}\times\cW^{i-1}$ to be
the channel $U_iV_i\to Y^nU^{i-1}V^{i-1}$, we see that
\begin{align*}
P_{(1)} &\text{ is } P^{-\dots--}; \\
P_{(2)} &\text{ is } P^{-\dots-+}; \\
P_{(3)} &\text{ is } P^{-\dots+-}; \\
&\vdots \\
P_{(n)} &\text{ is } P^{+\dots++}.
\end{align*}
It then follows from Theorem~\ref{thm:polar} that when $n$ is large, almost all channels $P_{(i)}$
are close to one of the five limiting channels. Also note that 
the $i$th channel assumes 
a genie that provides knowledge of the previous symbols 
$(U^{i-1}V^{i-1})$ at the receiver. 
This observation and Remark~\ref{rmk:limiting} motivate the following
coding scheme: Fix $\epsilon>0$, $\delta>0$. Let ${\cal
A}_U\subset(U_1,\dotsc,U_n)$ and ${\cal
A}_V\subset(V_1,\dotsc,V_n)$ denote
the sets of information symbols to be transmitted. Choose these sets
as follows:
\begin{itemize}[\labelwidth=2.5em]
\item[(i)] If $\|\cJ(P_{(i)})-(0,0,0)\|<\epsilon$ then $U_i\notin{\cal A}_U,V_i\notin{\cal A}_V$,
\item[(s.i)] if $\|\cJ(P_{(i)})-(0,1,1)\|<\epsilon$ then $U_i\notin{\cal A}_U,V_i\in{\cal A}_V$,
\item[(s.ii)] if $\|\cJ(P_{(i)})-(1,0,1)\|<\epsilon$ then $U_i\in{\cal A}_U,V_i\notin{\cal A}_V$,
\item[(s.iii)] if $\|\cJ(P_{(i)})-(1,1,1)\|<\epsilon$ then either $U_i\in{\cal A}_U,V_i\notin{\cal A}_V$, or $U_i\notin{\cal A}_U,V_i\in{\cal A}_V$,
\item[(s.iv)] if $\|\cJ(P_{(i)})-(1,1,2)\|<\epsilon$ then $U_i\in{\cal A}_U,V_i\in{\cal A}_V$,
\item[(s.v)] otherwise, $U_i\notin{\cal A}_U,V_i\notin{\cal A}_V$.
\end{itemize}
Choose the symbols in ${\cal A}_U^c$ and ${\cal A}_V^c$ independently and uniformly at random, and reveal their
values to the receiver. This choice of ${\cal A}_U$ and ${\cal A}_V$ ensures that all the information
symbols see `reliable' channels, provided that the previous symbols are decoded
correctly. Consequently, upon receiving $Y^n$, the receiver may attempt to decode the symbols
successively, in
the order $(U_1V_1)$, $(U_2V_2)\dotsc$,
and hope for a low block error probability. Furthermore, Theorem~\ref{thm:polar} and the 
preservation of $I^{\paren{12}}(P)$ throughout the recursive channel
splitting/combining process guarantee that for any choice of
$\epsilon$ and $\delta$ there exists $n_0$ such that $|{\cal
A}_U|+|{\cal A}_V|>n[I^{\paren{12}}(P)-\delta]$ whenever $n\ge n_0$.
This observation hints at the achievability of points on the dominant face of
$\cI(P)$. For a proof of achievability, it only remains to show that the block error probability
of the discussed scheme vanishes with increasing block length. We do
this next.

Let $\phi_i\colon\cY^n\times\fF_q^{i-1}\times\fF_q^{i-1}$,
$i=1,\dotsc,n$ denote the
ML decision rule for estimating $(U_iV_i)$ given $(Y^n,(UV)^{i-1})$.
Note that this corresponds to a \emph{genie aided} decision rule---the genie
provides $(UV)^{i-1}$---for
estimating $(U_iV_i)$ from the output $Y^n$. Let $E_i$ denote the event
$\phi_i(Y^n,(UV)^{i-1})\neq(U_iV_i)$. Observe that $E_i$
is precisely the error event of 
$P_{(i)}$. Now define a standalone
decoder, recursively through
$$
T_i=\phi_i(Y^n,T^{i-1}),\qquad i=1,\dotsc,n,
$$
and let $E_i'$ denote the event $T_i\neq(U_iV_i)$. Note that
$\cup_iE_i'$ is the block error event for the scheme discussed above,
and that
$$
\cup_i E_i = \cup_i E'_i.
$$
Hence, the block error probability can be bounded as
\begin{align}
\label{eqn:bit-error}
\Pr[\text{block error}]= \Pr [\cup_i E_i']
 = \Pr[\cup_i E_i]\le\sum_i\Pr[E_i]=\sum_iP_e(P_{(i)}).
\end{align}
Note that the transmission scheme
described above implies that the only non-zero error terms on the
right-hand-side of 
\eqref{eqn:bit-error} are those corresponding to the symbols in 
$\mathcal{A}_U$ and $\mathcal{A}_V$. We will show that almost all of
these terms are sufficiently small, i.e., that by removing a negligible fraction of information bits from
$\mathcal{A}_U$ and $\mathcal{A}_V$, the above sum can be made
to vanish. 
\begin{theorem}
\label{thm:rate}
For any $\beta<1/2$, the block error probability of the polar coding
scheme described above, under successive cancellation decoding, is 
$o(2^{-n^\beta})$.
\end{theorem}
Theorem~\ref{thm:rate} is an immediate corollary to
the following result. 

\begin{lemma}
\label{lem:rates}
For any $\epsilon>0$ and $\beta<1/2$,
\begin{itemize}[\labelwidth=2em]
\item[(r.1)]
$\lim_{\ell\to\infty}\frac1{2^\ell} \#\left\{ \bs\in\{-,+\}^\ell\colon
\|\cJ(P^\bs)-(0,1,1)\|<\epsilon, P_e(P^\bs[V])\ge
2^{-2^{\ell\beta}}\right\}=0$,
\item[(r.2)] $\lim_{\ell\to\infty}\frac1{2^\ell} \#\left\{ \bs\in\{-,+\}^\ell\colon
\|\cJ(P^\bs)-(1,0,1)\|<\epsilon, P_e(P^\bs[U])\ge
2^{-2^{\ell\beta}}\right\}=0$,
\item[(r.3)] $\lim_{\ell\to\infty}\frac1{2^\ell} \#\left\{ \bs\in\{-,+\}^\ell\colon
\|\cJ(P^\bs)-(1,1,1)\|<\epsilon, \max\{P_e(P^\bs[U\mid V]),P_e(P^\bs[V\mid
U])\}\ge
2^{-2^{\ell\beta}}\right\}=0$,
\item[(r.4)] $\lim_{\ell\to\infty}\frac1{2^\ell} \#\left\{ \bs\in\{-,+\}^\ell\colon
\|\cJ(P^\bs)-(1,1,2)\|<\epsilon, P_e(P^\bs[U])+P_e(P^\bs[V])\ge
2^{-2^{\ell\beta}}\right\}=0$.
\end{itemize}
\end{lemma}

The following proposition will be useful in the proof of
Lemma~\ref{lem:rates}.
\begin{proposition}
\label{prop:linear}
For all $\alpha,\gamma\in\fF_q$ and $\delta>0$,
\begin{align}
\lim_{\ell\to\infty}\frac1{2^\ell}\#\left\{\bs\in\{-,+\}^\ell\colon
I(P^\bs[\alpha,\gamma])\in(\delta,1-\delta)\right\}=0.
\end{align}
That is, the channels $\alpha U+\gamma V\to Y$ polarize to
become either perfect or useless. Moreover, convergence to perfect
channels is almost surely fast:
\begin{align}
\label{eqn:linear-rate}
\lim_{\ell\to\infty}\frac1{2^\ell}\#\left\{\bs\in\{-,+\}^\ell\colon
I(P^\bs[\alpha,\gamma])\ge 1-\delta,  Z(P^\bs[\alpha,\gamma])\ge
2^{-n^\beta}\right\}=0
\end{align}
for all $0<\beta,\delta<1/2$ and $\alpha,\gamma\in\fF_q$.
\end{proposition}
\begin{proof}
See Appendix B.
\end{proof}

\begin{newproof}{Proof of Lemma~\ref{lem:rates}}
The proof of (r.1) follows immediately from Proposition~\ref{prop:linear}
by taking $\alpha=0$ and $\gamma=1$, and by the relation
$P_e(P^\bs[V])\le qZ(P^\bs[V])$. Proofs of (r.2) and (r.4) follow
similarly.

To prove (r.3), we first observe that for any MAC $P$, 
\begin{align}
\label{eqn:error-linear}
\max\{P_e(P[U\mid V]),P_e(P[V\mid U])\}\le P_e(P[\alpha,\gamma])\le
qZ(P[\alpha,\gamma])
\end{align}
for all
$\alpha,\gamma\in\fF_q$. 
We also know from Proposition~\ref{prop:linear} that 
when $\ell$ is sufficiently
large, then for all $\alpha,\gamma\in\fF_q$,
\begin{align}
\label{eqn:l-condition}
I(P^\bs[\alpha,\gamma])\notin(o(\epsilon),1-o(\epsilon))
\end{align}
for almost all $\bs\in\{-,+\}^\ell$. It is an immediate consequence of
Lemma~\ref{lem:ugly} in Appendix C that
whenever \eqref{eqn:l-condition} is satisfied, then 
$$
\|\cJ(P^\bs)-(1,1,1)\|<\epsilon
\text{ implies }
I(P[\alpha,\gamma])>1-o(\epsilon)
$$
for some $\alpha,\gamma$. 
Claim (r.3) will then follow from
\eqref{eqn:linear-rate} and \eqref{eqn:error-linear}.
\end{newproof}

\section{Discussion}

The technique described above adapts the single-user polarization technique
of Ar{\i}kan to the two-user multiple-access channels.  It can be seen that
it retains the quality of being low complexity, and has similar error
probability scaling as the single-user case.  

As in the original polar code construction for single user channels,
the discussion for MAC polar codes above consider uniform input
distributions. How to achieve true channel capacity
with polar codes, using Gallager's method \cite[p.\ 208]{Gallager},
is discussed in \cite{STA2009}. 
The arguments in \cite[Section III.D]{STA2009} can easily be adapted
to multiple-access channels to extend the above results to rate regions
with non-uniform inputs.

A number of questions
for further study come to mind:
Unlike the single-user setting, where the `symmetric capacity' of a
channel is a single number, the `dominant face' of the set of rates that
correspond to uniformly distributed inputs is a line segment.  The
polarization technique outlined here does not in general achieve
the whole segment, but only a subset of it, for the simple reason that
the equations (1) and (2) are inequalities rather than equalities.  Is
there an alternative way to do MAC polarizaton and not suffer this loss?

A natural extension of the results presented is to the case of
multiple-access channels with more than
two users.  For such channels, one can fairly easily show that with a
similar construction as in this paper, there are a finite
number of limiting MACs, and that these extremal MACs have the property
that their rate regions are described by polymatroidal equations
with \emph{integer} right-hand sides, and are thus matroids.  One
encounters, however, a new phenomenon: not all matroids are
possible regions of a MAC.  The treatment of these require 
further techniques, which is the subject of \cite{AbbeTelatar2010}.

\section*{Appendix A}
In this section, we discuss how polar codes 
can be used to achieve arbitrary points in the capacity
region of any MAC with arbitrary number of users and discrete input
alphabets.
We follow the notation used in 
Section~\ref{sec:preliminaries}. 
For sake of simplicity, we show the achievability of corner 
points of $\cI(P)$ for a given $q$-ary input two-user MAC $P$, and
discuss how the result can be generalized.
\begin{theorem}
\label{thm:corner}
Let $P$ be a two-user $q$-ary input MAC. For any $\epsilon>0$ and $\beta<1/2$, there
exist two polar
codes $\mathcal{C}_1$ and $\mathcal{C}_2$ with sufficiently large
block lengths $n$, and with rates 
\begin{align*}
R_1&>I(X;Y)-\epsilon\\
R_2&>I(W;YX)-\epsilon
\end{align*}
such that if used by the two senders for transmission over $P$, their
average block error
probability does not exceed $2^{-n^\beta}$. This performance
is guaranteed under a receiver that decodes the \emph{messages
} successively.
\end{theorem}

\begin{proof}
Given a single-user $q$-ary input channel $Q$, let
$P_{e,n}(Q,\mathcal{A},u_{\mathcal{A}^c})$ denote the block error
probability of a polar code under successive cancellation (SC) decoding, with information set $\mathcal{A}$ and
frozen symbols fixed to $u_{\mathcal{A}^c}$, averaged over all 
messages. We know from \cite{Arikan2009} and \cite{STA2009}
that when $n$ is sufficiently large, there exists a set $\mathcal{A}$
with $|\mathcal{A}|>n(I(Q)-\epsilon)$ and
\begin{align}
\label{eqn:avgerr}
\frac1{2^{|\mathcal{A}^c|}}\sum_{u_{\mathcal{A}^c}}P_{e,n}(Q,\mathcal{A},u_{\mathcal{A}^c})=\mathcal{O}(2^{-n^\beta}).
\end{align}

Define two $q$-ary input channels $Q_1\colon\mathbb{F}_2\to\cY$
and $Q_2\colon\mathbb{F}_2\to\cY$ through the transition probabilities
\begin{align*}
Q_1(y\mid x)&=\tfrac1q\sum_w P(y\mid x,w),\\
Q_2(y,x\mid w)&=\tfrac1q P(y\mid x,w).
\end{align*}
Clearly, we have $I(Q_1)=I(X;Y)$ and $I(Q_2)=I(W;YX)$. Take $n$ sufficiently large and find sets $\mathcal{A}_1$ and
$\mathcal{A}_2$ with $|\mathcal{A}_1|>n(I(Q_1)-\epsilon)$ and
$|\mathcal{A}_2|>n(I(Q_2)-\epsilon)$, such that \eqref{eqn:avgerr}
holds when $(Q,\mathcal{A})$ is replaced with $(Q_1,\mathcal{A}_1)$ and
$(Q_2,\mathcal{A}_2)$, respectively. We will show
that the ensemble of polar code pairs characterized by $\mathcal{A}_1$
and $\mathcal{A}_2$ have small average error probability when used for
transmission over 
$P$. 

Consider a receiver that first makes an SC estimate
$\hat{X}^n=\phi_{X}(Y^n)$ on
the first sender's codeword $X^n$ based on the output $Y^n$, and
then produces
$\hat{W}^n=\phi_{W}(Y^n\hat{X}^n)$, 
where $\phi_{W}$ denotes the SC estimate of $W^n$ conditioned
on $(Y^nX^n)$. That is, the decoder for $W^n$ assumes that the decision on $X^n$ is always correct. 
The average block error probability of 
this scheme can be bounded using the relations
\begin{align*}
\Pr[\text{block error}]&=\Pr[\hat{X}^n\neq X^n \text{ or } \hat{W}^n\neq
W^n]\\
&=\Pr[\phi_{X}(Y)\neq X^n]+ \Pr[\phi_{W}(Y^n\hat{X}^n)\neq
W^n, \hat{X}^n=X^n]\\
&=\Pr[\phi_{X}(Y)\neq X^n]+ \Pr[\phi_{W}(Y^nX^n)\neq
W^n, \hat{X}^n=X^n]\\
&\le\Pr[\phi_{X}(Y)\neq X^n]+ \Pr[\phi_{W}(Y^nX^n)\neq W^n].
\end{align*}
The first probability term above can be written as
\begin{align*}
\Pr[\phi_X(Y^n)\neq X^n]
&=\frac1{q^n}\sum_{w^n}\Pr[\phi_X(Y^n)\neq X^n\mid W^n=w^n] \\
&=\frac1{q^{|\mathcal{A}_1^c|}}\sum_{u_{\mathcal{A}_1^c}}P_{e,n}(Q_1,\mathcal{A}_1,u_{\mathcal{A}_1^c})\\
&= \mathcal{O}(2^{-n^\beta}).
\end{align*}
Here, we obtained the second equality by observing that the codeword symbols
$X^n$ and $W^n$ are independent and uniformly distributed, which follows from the
uniform distribution on the frozen and information symbols. The
third inequality follows from \eqref{eqn:avgerr}. By the same line of
argument one can write
\begin{align*}
\Pr[\phi_W(Y^nX^n)\neq W^n]
&=\sum_{x_1^n} \Pr[\phi_W(Y^nX^n)\neq W_1^n,
X_1^n=x_1^n]\\
&=\frac1{q^{|\mathcal{A}_2^c|}}\sum_{u_{\mathcal{A}_2^c}}P_{e,n}(Q_2,\mathcal{A}_2,u_{\mathcal{A}_2^c})\\
&= \mathcal{O}(2^{-n^\beta}).
\end{align*}
Therefore, the block error probability, averaged over the ensemble of
polar code pairs is $\mathcal{O}(2^{-n^\beta})$. This lets us
conclude that there exists at least one pair of polar codes with
the 
promised rates and average block error probability.
\end{proof}

In \cite{RimUrb1996} and \cite{GRUW2001}, it was shown that any point
in the capacity region of an
$M$-user MAC can be expressed as a corner point of (at most) a $(2M-1)$-user MAC
rate region, possibly with non-uniform inputs. In addition, it
is shown in \cite[Section III]{STA2009} how
polar codes for non-binary channels can be used to achieve 
capacity of arbitrary discrete channels, by
inducing arbitrary 
non-uniform distributions on the input.
Modifying the above proof along these observations,
one can easily generalize Theorem~\ref{thm:corner} in order to show that polar codes achieve 
all points in the capacity region of any discrete input
MAC with arbitrary number of users.

\section*{Appendix B: Proof of Proposition~\ref{prop:linear}}
Given a channel $Q\colon\fF_q\to\cY$, define two channels $Q^b\colon\fF_q\to\cY^2$ and $Q^g\colon\fF_q\to\cY^2\times\fF_q$ through
\begin{align*}
Q^b(y_1,y_2\mid x_1) = \sum_{x_2} \tfrac12 Q(y_1\mid x_1+x_2) Q(y_2\mid x_2), \\
Q^g(y_1,y_2,x_1\mid x_2) = \tfrac12 Q(y_1\mid x_1+x_2) Q(y_2\mid x_2).
\end{align*}
It is easy to see that $I(Q^b)+I(Q^g)=2I(Q)$. We will show that for all $\alpha,\gamma\in\fF_q$,
\begin{itemize}
\item[(i)]
$P[\alpha,\gamma]^g$ is degraded with respect to $P^+[\alpha,\gamma]$,
\item[(ii)]
$P[\alpha,\gamma]^b$ is equivalent to $P^+[\alpha,\gamma]$,
\end{itemize}
implying 
$$
I(P^+[\alpha,\gamma])+I(P^-[\alpha,\gamma])\ge 2I(P[\alpha,\gamma]).
$$
This, in addition to (i), (ii), and Lemma~\ref{lem:basic}, implies the
convergence of the channels $P[\alpha,\gamma]$ to extremals---the
proof is identical to that of Corollary~\ref{cor:i1}. That is,
$$
\lim_{\ell\to\infty}\frac1{2^\ell}\#\left\{\bs\in\{-,+\}^\ell\colon I(P^\bs[\alpha,\gamma])\in(\delta,1-\delta)\right\}=0.
$$
To prove the claim on the rate of convergence, we will show that 
\begin{align}
\label{eqn:z-alpha}
Z(P^-[\alpha,\gamma])\le 2Z(P[\alpha,\gamma])\quad\text{ and }\quad
Z(P^+[\alpha,\gamma])\le qZ(P[\alpha,\gamma])^2. 
\end{align}
The proof will then follow from previous results, namely
\begin{lemma}[\cite{ArikanTelatar2009},\cite{STA2009}]
\label{lem:rate}
For any $q$-ary input channel $Q\colon\mathbb{F}_2\to\cY$,
channels $Q^b$ and $Q^g$ satisfy
\begin{align}
\label{eqn:z-evolution}
Z(Q^b)\leq 2Z(Q)\quad\text{ and }\quad Z(Q^g)\le qZ(Q)^2.
\end{align}
In particular, this implies that
\begin{align}
\label{eqn:rate}
\lim_{\ell\to\infty} \frac1{2^\ell}\# \left\{\bs\in\{g,b\}^\ell\colon
	 I(Q^\bs)>1-\epsilon, Z(Q^\bs) > 2^{-2^{\ell\beta}}\right\}=0
\end{align}
for all $0<\epsilon,\beta<1/2$.
\end{lemma}
It thus remains to show (i), (ii), and \eqref{eqn:z-alpha}.

\begin{newproof}{Proof of (i)}
We have by definition
\begin{align}
\notag
P^+[\alpha,\gamma](y_1,y_2,u_1,v_1\mid s)&=
	\sum_{\substack{u_2,v_2:\\ \alpha u_2+\gamma v_2=s}}
	\frac1q P^+(y_1,y_2,u_1,v_1\mid u_2,v_2)\\
\label{eqn:p-plus}
&=\sum_{\substack{u_2,v_2:\\ \alpha u_2+\gamma v_2=s}}
	\tfrac1{q^3}P(y_1\mid u_1+u_2,v_1+v_2)P(y_2\mid u_2,v_2).
\end{align}
On the other hand,
\begin{align*}
P[\alpha,\gamma]^g(y_1,y_2,x\mid s)&=
	\tfrac1q P[\alpha,\gamma](y_1\mid x+s) 
	P[\alpha,\gamma](y_2\mid s)\\
&=\tfrac1{q^3} \sum_{\substack{u_1,v_1,u_2,v_2:\\ \alpha u_1+\gamma v_1=x+s
\\  \alpha u_2+\gamma v_2=s}}
	P(y_1\mid u_1,v_1)P(y_2\mid u_2,v_2)
\end{align*}
Since the constraints $\alpha u_1+\gamma v_1=x+s$ and $\alpha
u_2+\gamma v_2=s$ are linear, the above sum can be rewritten as
\begin{align}
\label{eqn:p-g}
P[\alpha,\gamma]^g(y_1,y_2,u_1,v_1\mid s)=
	\tfrac1{q^3} 
	\sum_{\substack{u_1,v_1,u_2,v_2:\\ \alpha u_1+\gamma v_1=x
\\  \alpha u_2+\gamma v_2=s}}
	P(y_1\mid u_1+u_2,v_1+v_2)P(y_2\mid u_2,v_2).
\end{align}
Comparing \eqref{eqn:p-plus} and \eqref{eqn:p-g}, we observe that the
channel $P[\alpha,\gamma]^g$ is obtained by processing the output
$(Y_1,Y_2,U_1,V_1)$ of $P^+[\alpha,\gamma]$ to retain $(Y_1,Y_2,\alpha
U_1+\gamma V_1)$. This completes the proof.
\end{newproof}

\begin{newproof}{Proof of (ii)}
We have
\begin{align}
\notag
P^-[\alpha,\gamma](y_1,y_2,\mid x)&=
	\sum_{\substack{u_1,v_1:\\ \alpha u_1+\gamma v_1=x}}
	\frac1q P^-(y_1,y_2\mid u_1,v_1)\\
\label{eqn:p-minus}
&=\sum_{\substack{u_1,v_1,u_2,v_2:\\ \alpha u_1+\gamma v_1=x}}
	\tfrac1{q^3}P(y_1\mid u_1+u_2,v_1+v_2)P(y_2\mid u_2,v_2).
\end{align}
On the other hand,
\begin{align*}
P[\alpha,\gamma]^b(y_1,y_2\mid x)&=
	\tfrac1q \sum_sP[\alpha,\gamma](y_1\mid x+s) 
	P[\alpha,\gamma](y_2\mid s)\\
&= \sum_{\substack{u_1,v_1,u_2,v_2:\\ \alpha u_1+\gamma v_1=x+s
\\  \alpha u_2+\gamma v_2=s}}
	\sum_s
	\tfrac1{q^3} P(y_1\mid u_1,v_1)P(y_2\mid u_2,v_2)
\end{align*}
As in the proof of (i), we can rewrite the above sum as 
\begin{align}
\notag
P[\alpha,\gamma]^b(y_1,y_2,u_1,v_1\mid s)&=
	\sum_{\substack{u_1,v_1,u_2,v_2:\\ \alpha u_1+\gamma v_1=x
\\  \alpha u_2+\gamma v_2=s}}
	\sum_s
	\tfrac1{q^3} P(y_1\mid u_1+u_2,v_1+v_2)P(y_2\mid u_2,v_2)\\
\label{eqn:p-b}
&=
	\sum_{\substack{u_1,v_1,u_2,v_2:\\ \alpha u_1+\gamma v_1=x}}
	\tfrac1{q^3} P(y_1\mid u_1+u_2,v_1+v_2)P(y_2\mid u_2,v_2).
\end{align}
Comparing \eqref{eqn:p-minus} and \eqref{eqn:p-b} we conclude that
$P[\alpha,\gamma]^b$ and $P^-[\alpha,\gamma]$ are equivalent.
\end{newproof}

\begin{newproof}{Proof of \eqref{eqn:z-alpha}}
It immediately follows from \eqref{eqn:z-evolution} and (ii) that $Z(P^b[\alpha,\gamma])\le
2Z(P[\alpha,\gamma])$. 
In order to complete the proof, we will show that $Z(P^+[\alpha,\gamma])\le
Z(P[\alpha,\gamma]^g)$. It will then
follow from \eqref{eqn:z-evolution} that $Z(P^+[\alpha,\gamma])\le qZ(P[\alpha,\gamma])^2$.

Define the channels
\begin{align*}
P^+_{u_1v_1}[\alpha,\gamma](y_1,y_2\mid
s)&=q^2P^+[\alpha,\gamma](y_1,y_2,u_1,u_2\mid s),\\
P_x[\alpha,\gamma]^g(y_1,y_2\mid s)&=qP[\alpha,\gamma]^g(y_1,y_2,x\mid
s).
\end{align*}
An inspection of \eqref{eqn:p-plus} and \eqref{eqn:p-g} reveals that
$$
P_x[\alpha,\gamma]^g(y_1,y_2\mid s)=\sum_{\alpha u_1+\gamma v_1=x}\tfrac1qP^+_{u_1v_1}[\alpha,\gamma](y_1,y_2\mid
s).
$$
Also, we clearly have
\begin{align*}
Z(P^+[\alpha,\gamma])&=\tfrac1{q^2}\sum_{u_1,v_1}Z(P^+_{u_1v_1}[\alpha,\gamma]),\\
Z(P[\alpha,\gamma]^g)&=\tfrac1q\sum_xZ(P_x[\alpha,\gamma]^g).
\end{align*}
It then follows from the concavity of the Bhattacharyya parameter in
the channel (cf.\ Lemma~\ref{lem:z-convex} below) that
\begin{align*}
Z(P[\alpha,\gamma]^g)&=\tfrac1q\sum_xZ(P_x[\alpha,\gamma]^g)\\
&=\tfrac1q\sum_xZ\left(\sum_{\alpha u_1+\gamma
v_1=x}\tfrac1qP^+_{u_1v_1}[\alpha,\gamma]\right)\\
&\ge \tfrac1{q^2}\sum_x\sum_{\alpha u_1+\gamma
v_1=x}Z(P^+_{u_1v_1}[\alpha,\gamma])\\
&=\tfrac1{q^2}\sum_{u_1,
v_1}Z(P^+_{u_1v_1}[\alpha,\gamma])\\
&=Z(P^+[\alpha,\gamma]),
\end{align*}
completing the proof.
\end{newproof}

\begin{lemma}
\label{lem:z-convex}
Let $Q,Q_1,\dotsc,Q_K$ be $q$-ary input channels with
$$
Q=\sum_{k=1}^K p_kQ_k,
$$
where $p_k\ge0$ and $\sum_kp_k=1$.
Then,
$$
Z(Q)\geq \sum_{k=1}^K p_kZ(Q_k).
$$
\end{lemma}

\begin{proof}
The proof is identical to that of \cite[Lemma 4]{Arikan2009}: 
\begin{align*}
Z(Q)&=\frac1{q-1}\left[-1+\frac1q\sum_y\left(\sum_x\sqrt{Q(y\mid
x)}\right)^2\right]\\
&\ge\frac1{q-1}\left[-1+\frac1q\sum_y\sum_kp_k\left(\sum_x\sqrt{Q_k(y\mid
x)}\right)^2\right]\\
&=\sum_kp_kZ(Q_k).
\end{align*}
Here, the inequality follows from \cite[p.\ 524, ineq.\
(h)]{Gallager}.
\end{proof}
\section*{Appendix C}
\begin{lemma}
\label{lem:ugly}
Let $X,W\in\fF_q$ be independent and uniformly distributed random
variables, and let $Y$ be an arbitrary random variable. For every
$\epsilon>0$, there exists $\delta>0$ such that 
\begin{itemize}
\item[(i)]
$I(X;Y)<\delta$, $I(W;Y)<\delta$, $H(X\mid YW)<\delta$, $H(W\mid
YX)<\delta$ and
\item[(ii)]
$H(\alpha X+\gamma W\mid Y)\notin(\delta,1-\delta)$ for all
$\alpha,\gamma\in\fF_q$,
\end{itemize}
implies 
$$
I(\alpha'X+\gamma'W;Y)>1-\epsilon.
$$
for some $\alpha',\gamma'\in\fF_q$.
\end{lemma}

\begin{proof}
Let $\pi$ be a permutation on $\fF_q$, and let
$$
p_\pi(x,w)=\begin{cases}\tfrac1q & \text{if }w=\pi(x)\\ 0
&\text{otherwise}\end{cases}.
$$
Note that $H(X)=H(W)=1$ and $H(W\mid X)=H(X\mid W)=0$ whenever $(X,W)$
is distributed as $p_\pi$.
We claim that for every $\pi$, there exist 
$\alpha_\pi,\gamma_\pi\in\fF_q\backslash\{0\}$ such that 
$$
H(\alpha_\pi X+\gamma_\pi W)<1-c(q),
$$
where $c(q)>0$ depends only on $q$. To see this, given a permutation
$\pi$, let
\begin{align}
\label{eqn:alpha-gamma}
\alpha_\pi:=\pi(0)-\pi(1),\quad\gamma_\pi:=1,\quad\mu:=\pi(0).
\end{align}
Clearly, $\alpha_\pi\neq0$. It is also easy
to check that with these definitions we have
$$
\Pr[\alpha_\pi X+\gamma_\pi W=\mu]\ge\Pr[(X,W)=(0,\pi(0))]+
\Pr[(X,W)=(1,\pi(1))]= \tfrac2q,
$$
which yields the claim. It also follows from the continuity of entropy
in the $L_1$ metric that 
$$
\|p_{XW}-p_\pi\|\le\epsilon\quad\text{ implies }\quad H(\alpha_\pi
X+\gamma_\pi W)\le (1-c(q))(1-o(\epsilon))^{-1}.
$$

We now show that for every $\epsilon>0$, there exists a $\delta>0$
such that whenever $H(W|YX)<\delta$, $H(X|YW)<\delta$, $I(W;Y)<\delta$ and
$I(X;Y)<\delta$, then there is a set $S$ of $y$'s with $p_Y(S)>1-\epsilon$ such that for all $y\in S$
$$
\min_\pi\|p_{XW\mid Y=y}-p_\pi\|_1<\epsilon.
$$

From $I(W;Y)<\delta$, Pinsker's
inequality yields
$$
\sum_y p_Y(y)\|\tfrac1q-p_{W|Y=y}\|_1 < \sqrt{2\delta\ln 2}<2\sqrt{\delta},
$$
and we conclude that the set
$$
G:=\{y:\|\tfrac1q-p_{W|Y=y}\|_1<\delta^{1/4}\}
$$
has probability at least $1-2\delta^{1/4}$.  Note that for $y\in G$,
\begin{align*}
\tfrac1q-\tfrac1q\delta^{1/4}&< p_{W|Y=y}(0|y)<\tfrac1q+\tfrac1q\delta^{1/4}\\
\tfrac1q-\tfrac1q\delta^{1/4}&< p_{W|Y=y}(1|y)<\tfrac1q+\tfrac1q\delta^{1/4}
\end{align*}
and thus
\begin{align*}
(1-\delta^{1/4})p_Y(y)&< p_{Y|W}(y|0)< (1+\delta^{1/4})p_Y(y)\\
(1-\delta^{1/4})p_Y(y)&< p_{Y|W}(y|1)< (1+\delta^{1/4})p_Y(y)
\end{align*}
Furthermore, as
$$
\delta>H(X|WY)=\sum_{w,y}p_{WY}(w,y)H(X|W=w,Y=y),
$$
the set $\{(w,y)\colon H(X|W=w,Y=y)>\sqrt{\delta}\}$ has probability
at most $\sqrt{\delta}$.  Let
$$
B_w=\{y\colon H(X|W=w,Y=y)>\sqrt{\delta}\},\quad w\in\fF_q\text{ and }
B=\cup_w B_w.
$$
Then,
\begin{align*}
P_Y(G\cap B_w)&=\sum_{y\in G\cap B_w}p_Y(y)\\
&\leq[1-\delta^{1/4}]^{-1}\sum_{y\in G\cap B_w}p_{Y|W}(y|w)\\
&\leq[1-\delta^{1/4}]^{-1}\sum_{y\in B_w}p_{Y|W}(y|w)\\
&\leq[1-\delta^{1/4}]^{-1}2\sum_{y\in B_w}p_{WY}(w,y)\\
&\leq[1-\delta^{1/4}]^{-1}2\sqrt{\delta}
\end{align*}
for all $w\in\fF_q$, and thus
$$
P_Y(G\cap B)\leq 2q\sqrt{\delta}[1-\delta^{1/4}]^{-1},
$$
and the set $S=G\cap B^c$ has probability
$$
P_Y(S)>1-2\delta^{1/4}-2q\sqrt{\delta}/[1-\delta^{1/4}]=1-o(\delta).
$$
Note that for all $y\in S$ we have for any $w$, $|\tfrac1q-p_{W|Y=y}(w)|<o(\delta)$,
and $p_{X|WY}(x|w,y)\not\in (o(\delta),1-o(\delta))$, and thus
$$
\min_\pi\|p_{WX|Y=y}-p_\pi\|<o(\delta).
$$
In particular, this implies that there exist $\pi'$ and $S'\subset S$ with
$P_Y(S')\ge P_Y(S)/q!$ such that
$$
\|p_{WX|Y=y}-p_{\pi'}\|<o(\delta)
$$
for all $y\in S'$.
Letting $\alpha'=\alpha_{\pi'}\text{ and }\gamma'=\gamma_{\pi'},$
where $\alpha_{\pi'}$ and $\gamma_{\pi'}$ are defined as in
\eqref{eqn:alpha-gamma}, we obtain 
\begin{align*}
H(\alpha'X+\gamma'W\mid Y)&\le
	P_Y(S')(1-c(q))(1-o(\epsilon))^{-1}+P_Y(S'^c)\\
&=(1-c_2)(1-o(\epsilon))^{-1}
\end{align*}
where $c_2>0$ depends only on $q$. Noting that
$I(\alpha'X+\gamma'W;Y)\notin(\delta,1-\delta)$ by assumption, and
that
\begin{align*}
I(\alpha'X+\gamma'W;Y)&=H(\alpha'X+\gamma'W)-H(\alpha'X+\gamma'W\mid
Y)\\
&\ge 1-(1-c_2)(1-o(\epsilon))^{-1},
\end{align*}
we see that if $\delta$ is sufficiently small, then 
$I(\alpha'X+\gamma'W;Y)\ge1-\delta$.
\end{proof}

\end{document}

%% file: macs.tex
\setlength{\unitlength}{1bp}
\begin{picture}(286,210)
\small
\put(0,125){
	\put(0,0){\includegraphics{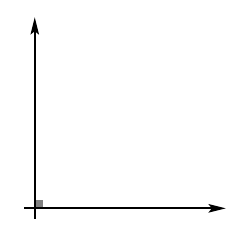}}
	\put(10,1){\makebox(0,0){$0$}}
%	\put(55,1){\makebox(0,0){$1$}}
	\put(2,10){\makebox(0,0){$0$}}
%	\put(2,55){\makebox(0,0){$1$}}
	\put(74,9){\makebox(0,0){$R_1$}}
	\put(10,72){\makebox(0,0){$R_2$}}
	\put(37,-12){\makebox(0,0){(000)}}
}
\put(105,125){
	\put(0,0){\includegraphics{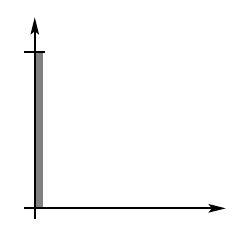}}
	\put(10,1){\makebox(0,0){$0$}}
%	\put(55,1){\makebox(0,0){$1$}}
	\put(2,10){\makebox(0,0){$0$}}
	\put(2,55){\makebox(0,0){$1$}}
	\put(74,9){\makebox(0,0){$R_1$}}
	\put(10,72){\makebox(0,0){$R_2$}}
	\put(37,-12){\makebox(0,0){(011)}}
}
\put(210,125){
	\put(0,0){\includegraphics{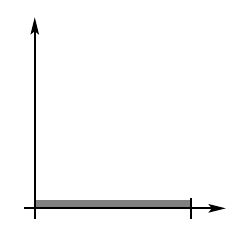}}
	\put(10,1){\makebox(0,0){$0$}}
	\put(55,1){\makebox(0,0){$1$}}
	\put(2,10){\makebox(0,0){$0$}}
%	\put(2,55){\makebox(0,0){$1$}}
	\put(74,9){\makebox(0,0){$R_1$}}
	\put(10,72){\makebox(0,0){$R_2$}}
	\put(37,-12){\makebox(0,0){(101)}}
}
\put(55,20){
	\put(0,0){\includegraphics{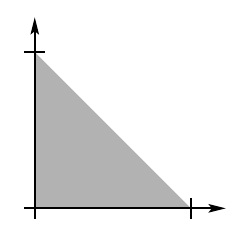}}
	\put(10,1){\makebox(0,0){$0$}}
	\put(55,1){\makebox(0,0){$1$}}
	\put(2,10){\makebox(0,0){$0$}}
	\put(2,55){\makebox(0,0){$1$}}
	\put(74,9){\makebox(0,0){$R_1$}}
	\put(10,72){\makebox(0,0){$R_2$}}
	\put(37,-12){\makebox(0,0){(111)}}
}
\put(160,20){
	\put(0,0){\includegraphics{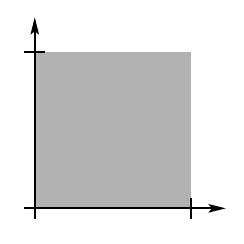}}
	\put(10,1){\makebox(0,0){$0$}}
	\put(55,1){\makebox(0,0){$1$}}
	\put(2,10){\makebox(0,0){$0$}}
	\put(2,55){\makebox(0,0){$1$}}
	\put(74,9){\makebox(0,0){$R_1$}}
	\put(10,72){\makebox(0,0){$R_2$}}
	\put(37,-12){\makebox(0,0){(112)}}
}
\end{picture}